\documentclass[onecolumn, referee]{svjour3}         

\usepackage{mathrsfs,amssymb,amsmath,amsfonts,latexsym,graphicx,fullpage}

\usepackage{color} 
\renewcommand{\theenumi}{\roman{enumi}}

\newcommand{\lastcorrections}%
{{
 \begin{sloppypar}
    \baselineskip -0.2in
    \tiny\bf\noindent
last corrections:\\
\end{sloppypar}
}}

\newcommand{\margincomment}[1]%
    {{%
      \marginpar{{\tiny\begin{minipage}{0.5in}
                       \begin{flushleft}
                          {#1}
                       \end{flushleft}
                       \end{minipage}
                }}
    }}

\newcommand{\ignore}[1]{}

\newcommand{\braced}[1]{{ \left\{ #1 \right\} }}

\newcommand{\NP}{{\mathbb{NP}}}

\newcommand{\OPT}{\textsc{Opt}}
\newcommand{\delay}{{\textit{delay}}}
\newcommand{\excess}{E}
\newcommand{\rmS}{\textrm{S}}
\newcommand{\rmX}{\textrm{X}}
\newcommand{\rmF}{\textrm{F}}
\newcommand{\rmC}{\textrm{C}}

\newtheorem{Theorem}{\em{Theorem}}
\newtheorem{Lemma}{\em{Lemma}}

\usepackage[hyphens]{url}
\usepackage{algorithm}
\usepackage{algorithmic}
\usepackage{amssymb, amsmath}
\usepackage{mathrsfs}
\usepackage{natbib}

\newcommand{\pmtn}{{\textit{pmtn}}}
\newcommand{\ThreePartition}{{\textsc{3-Partition}}}

\begin{document}

\title{A Note on $\NP$-Hardness of Preemptive Mean Flow-Time \\Scheduling for Parallel Machines}


\author{Odile  Bellenguez-Morineau \and  Marek Chrobak  \and Christoph
  D\"urr \and  Damien Prot} \institute{Odile  Bellenguez-Morineau \and
  Damien Prot  \at LUNAM Universit\'e, \'{E}cole des  Mines de Nantes,
  IRCCyN UMR CNRS  6597 (Institut de Recherche en  Communication et en
  Cybern\'{e}tique de  Nantes), 4  rue Alfred Kastler,  La Chantrerie,
  BP20722, 44307 Nantes Cedex 3, France
  \\
  \email{odile.morineau, damien.prot@mines-nantes.fr}\\
  Marek Chrobak \at Computer Science Department, University of
  California, Riverside, CA 92521, USA\\
  \email{marek@cs.ucr.edu}\\
  Christoph D\"urr \at CNRS, LIP6, Universit\'e Pierre et Marie Curie,
  4 place Jussieu,
  75252 Paris Cedex 05, France\\
  \email{Christoph.Durr@lip6.fr}\\
}

\date{Received: date / Accepted: date}

\maketitle

\begin{abstract} 
In the paper \emph{``The complexity of mean flow time scheduling problems with release times''},
by Baptiste, Brucker, Chrobak, D\"{u}rr, Kravchenko and Sourd, the authors
claimed to prove strong $\NP$-hardness of the scheduling problem
$P|\pmtn,r_j|\sum C_j$, namely multiprocessor preemptive scheduling where the objective is to
minimize the mean
flow time. We point out a serious error in their proof and give a new proof of strong
$\NP$-hardness for this problem.
\end{abstract}


\section{Introduction}
\label{sec: introduction}

In~\cite{BBCDKS07} the following scheduling problem was considered. We
are  given $N$  jobs, where each  job $j$  has some  release time  $r_j$ and
processing  time $p_j$, all positive integers.  
The  goal is  to preemptively  schedule these
jobs  on  $m$  parallel  identical  machines,  so as to  minimize  the mean flow
time (or, equivalently, the total completion time).  
We use a standard definition of preemptive schedules, namely
 a schedule is specified by assigning
to each job $j$ a finite set of execution  intervals of $j$, where each such interval
is associated with some machine $k$. For a schedule to be feasible,
all intervals  associated with any machine $k$ must be disjoint, and all
intervals  assigned to the  same job $j$ must be disjoint and start not
earlier than at time $r_j$.
The completion time of a job $j$, denoted by $C_j$, is defined as 
the right endpoint of the last execution  interval  assigned  to $j$.

In the three-field notation for scheduling problems, 
this problem is denoted  $P|\pmtn,r_j|\sum C_j$.  The special case when
there  are  only  $m=2$  parallel   machines  has  been  shown  to  be
$\NP$-hard in \cite{DLY90}, while the single machine variant is
solvable in polynomial time, see~\cite{B74}.

The  computational complexity  of $P|\pmtn,r_j|\sum  C_j$  was studied
by~\cite{BBCDKS07}, and one  of the results in that  paper was a proof
of  strong $\NP$-hardness.   Unfortunately,  as we  show  in the  next
section, that proof  has a serious flaw.  Therefore we provide a new proof
of strong $\NP$-hardness in Section~\ref{sec: new proof}, which builds
on the overall structure of the proof from~\cite{BK04}.


\section{Counterexample to the Proof in ~\cite{BBCDKS07}}
\label{sec: counterexample}

The  (faulty) strong $\NP$-hardness proof  of~\cite{BBCDKS07}  is based on  a
reduction from {\ThreePartition}. It converts an instance of {\ThreePartition}
into a collection of jobs of three types: x-jobs, B-jobs and 1-jobs.  
The role of the B- and 1-jobs
is to force any optimal schedule  to start these jobs at their release
time, and leave  only a small time interval available for the  x-jobs. The key idea
of the proof was that in  an optimal schedule, the x-jobs would have
to be scheduled  without preemption in this interval,
with three jobs per machine, and in this way the resulting schedule would
represent a solution of the original instance of {\ThreePartition}.

The error in this proof is  that it is always possible to schedule
the  x-jobs in the allowed interval using the method described
by~\cite{M59} (see Theorem~3.1):
Order the jobs $1,2,...,N$ arbitrarily. Using the unit slots of
the jobs, in this order, fill the
allowed interval for machine $1$, going from left to right, then
fill the allowed interval for machine $2$, and so on, until all jobs
are processed. 

We now show a specific counter-example to the construction of~\cite{BBCDKS07}.
Let the instance of {\ThreePartition} consist of
positive integers $x_1,\dots,x_{3n},y$  that satisfy $\sum_{i=1}^{3n}x_i = ny$
and  $\frac{y}{4} <  x_i  <  \frac{y}{2}$ for  each  $i$. The objective is to
decide if there  exist a partition of $\{1,...,3n\}$ into
$n$ sets $P_1,...,P_n$ such that $\sum_{i\in P_k}x_i = y$ for all $k$.

Let $A=6ny$ and $B=18n^2y^2$. The reduction in~\cite{BBCDKS07} converts the above
instance of {\ThreePartition} into an instance of $P|\pmtn,r_j|\sum C_j$ with
$n$ machines and $N=4n+An$ jobs of three types:
\begin{itemize}
\item
x-jobs $j$: for all $j\in\{1,\dots,3n\}$, with $r_j=0$ and $p_j=Ax_j$,
\item
B-jobs $j$: for all $j\in \{3n+1,\dots,4n\}$, 
with $r_j=Ay$ and $p_j=B$, 
\item
1-jobs $j$: for all $j\in \{4n+1,\dots,N\}$, 
with $r_j=Ay+B$ and $p_j=1$.
\end{itemize}
The authors claim that the instance of {\ThreePartition} has a solution if and only if there exists 
a schedule with $\sum_{j=1}^N C_j \leq D$, where $D=3nAy+n(Ay+B)+n\sum_{i=1}^A(Ay+B+i)$.
The ``only if'' part of this claim is easy:
on a machine $k$, we first schedule the three x-jobs $j\in P_k$,
then one B-job and finally $A$ 1-jobs. As the example below shows, however, the
``if'' implication is not valid.


\begin{figure}[ht]
\begin{center}
\includegraphics[width=5.5in]{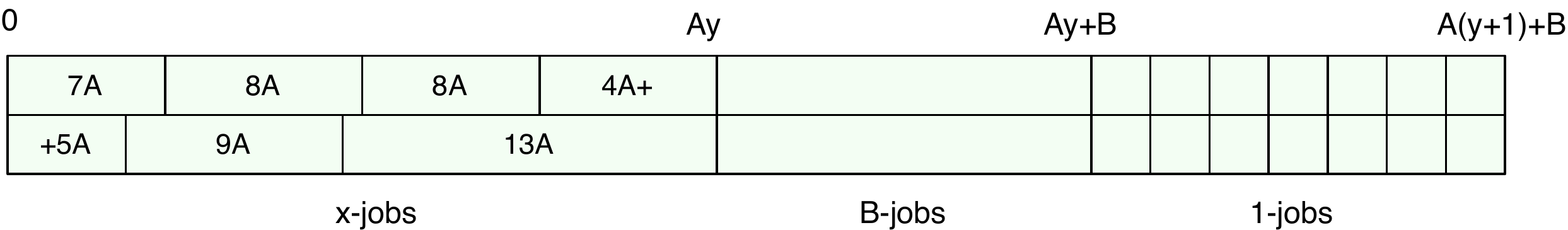}
\caption{The schedule of the jobs from our counter-example. (Picture is not to scale.) 
The x-jobs are scheduled using
McNaughton's algorithm. The 4th x-job is scheduled in two parts, one of length $4A$ scheduled
on machine 1 and the other, of length $5A$, on machine 2.}
\label{fig: counterexample}
\end{center}
\end{figure}


Our counter-example is an instance of {\ThreePartition} with $n=2$ and
$x_1=7$, $x_2=8$, $x_3=8$, $x_4=9$, $x_5=9$, $x_6=13$ and $y=27$.
Clearly, there is no solution to {\ThreePartition} since the partition
set containing $x_6$
would need two additional numbers that add up to $27-13=14$, which is not possible.
In the corresponding instance of $P|\pmtn,r_j|\sum C_j$ we will have 
$A = 324$ and $B = 52488$. The instance consists of $n=2$ machines and
$656$ jobs:
\begin{itemize}
\item
Six x-jobs with processing times $7A$,  $8A$, $8A$,  $9A$, $9A$  and $13A$, all
	released at time 0,
\item
Two B-jobs with processing time $B$, released at time $27A$, and
\item
$2A$ 1-jobs with processing time $1$, released at time $27A+B$.
\end{itemize}
As shown in Figure~\ref{fig: counterexample},
all x-jobs in this instance can be scheduled in the time interval $[0,27A]$ using
McNaughton's method, and the
objective value of the shown schedule is at most $D=3nAy+n(Ay+B)+n\sum_{i=1}^A(Ay+B+i)$,
because all x-jobs complete no later than at time $Ay$.
Thus the ``if'' implication does not hold.


\section{A New Proof of $\NP$-Hardness}
\label{sec: new proof}

In this section we present a corrected proof. Our proof, as before, uses a reduction from
$\ThreePartition$, although the construction is more involved.

It has been proven that, for any instance of $P|\pmtn,r_j|\sum C_j$,
there exists an optimal schedule
where preemptions occur only at integer times (see~\cite{BBCDKS07}, for example);
hence we will make this assumption throughout the paper, namely we will assume
that the time is divided into unit time slots, each either idle or fully
filled by a unit fragment of one job.

Let us fix an instance of  $\ThreePartition$ consisting of positive integers
$x_1,\dots,x_{3n},y$, where $\sum_i x_i=ny$ and $y/4 <  x_i < y/2$ for all $i$. Without
loss of generality, we assume that $n\ge 2$. Using this instance,
we construct an instance of $P|\pmtn,r_j|\sum C_j$ with $n$ machines. In this construction
we use the following values:
\begin{align*}
\lambda &= 2n(n-1)
\\
L &= n\lambda y+\lambda
\\ 
{\OPT} &= yL\frac{n(n-1)}{2} + nyL\frac{n(n-1)}{2}  + n \lambda y
\\
T &= n^2 yL + 3n{\OPT}
\end{align*}
Let $B_i=\sum_{l=1}^{i-1}(nLx_l+{\OPT})$ for $i = 1,...,3n$. Note that
$B_1 = 0$. For convenience, we also let $B_{3n+1} = T$.
We partition the time interval $[0,T)$ into $3n$
\emph{blocks}, where the $i$-th block is the interval $[B_i,B_{i+1})$.
Thus each block $i$, has length $nLx_i+{\OPT}$.

We also have another special ``cork'' block in the interval $[T,T+{\OPT})$. 
The instance will have jobs of four types:
\begin{description}
\item{$\rmS$-\emph{jobs}:}  In  each  block  $i$, for each integer time 
 $t =  B_i,B_i+1,...,B_i+(n-1)Lx_i-1$, we release $n-\lceil  \frac{t-B_i+1}{Lx_i}\rceil$
 jobs with unit
 processing time.  The idea is that these S-jobs should form a staircase-shaped
 schedule in their block $i$, with  each  stair step having (ideally) length $Lx_i$.
\item{$\rmX$-\emph{jobs}:}  These jobs  correspond   to    the   numbers $x_i$ in
   the original {\ThreePartition} instance: in each  block $i$, we will have
	one job $X_i$ of length
  $\lambda x_i$  released  at the beginning  of the block,  i.e. at time $B_i$.
\item{$\rmF$-\emph{jobs}:} This is a set of 
	$n$ jobs $F_1,...,F_n$, released at time $0$. Their role is to fill
  the  idle times  in each  machine $k$.  The length  of each job  $F_k$ is
  $T-(k-1)Lny - \lambda y$.
\item{$\rmC$-\emph{jobs}:} This is a set of $n{\OPT}$  unit jobs 
  released in the cork block. Specifically, for each time $t = T,...,T+{\OPT}-1$, we
 	release $n$ unit jobs at time $t$. The purpose of these jobs is to force
	all $\rmS$-jobs, $\rmX$-jobs and $\rmF$-jobs to complete no later than at time $T$.
\end{description}

Note that this transformation can be computed in polynomial time if the
instance of $\ThreePartition$ and the constructed instance of $P|\pmtn,r_j|\sum C_j$
are represented in the unary encoding. 
This will be sufficient for our purpose, since $\ThreePartition$ is strongly $\NP$-hard.

To simplify calculations,
instead of  minimizing $\sum C_j$,  we  will use the objective function
$\sum D_j$ where $D_j=C_j-r_j-p_j$, which is of course equivalent.
We will refer to $D_j$  as  the  \emph{delay}  of  job  $j$. For a schedule $\sigma$,
by $\delay(\sigma)$ we will denote the total delay (that is, $\sum D_j$)
of $\sigma$. For $\phi\in\braced{\rmS,\rmX,\rmF,\rmC}$, by
$\delay_{\phi}(\sigma)$ we denote the contribution of $\phi$-jobs to the
total delay in $\sigma$.


\begin{Theorem} \label{thm:main}
The instance of {\ThreePartition} has a solution if and only if
the instance of $P|\pmtn,r_j|\sum C_j$  constructed above
has a schedule with total delay at most ${\OPT}$.
\end{Theorem}

The rest of this section is devoted to the proof of Theorem~\ref{thm:main}.
We will prove the two implications in the theorem separately.

\medskip
$(\Rightarrow)$
For  the  ``only if'' direction,  consider  a  partition
$P_1,\dots,P_n$  such  that  $\sum_{i\in  P_k}x_i=y$  for  every  $k$.
Schedule all $\rmS$-jobs at their release times to form stairs in each block
$i$. Schedule the $\rmC$-jobs at  their release times, to form the cork block. 
For each $k$, if $i\in P_k$, then schedule the $\rmX$-job corresponding to $x_i$ 
at offset $(k-1)Lx_i$ in block $i$ on machine $k$. 
These jobs are scheduled  without preemption.  
For each $k$, schedule job  $F_k$ preemptively on
machine $k$, so  that it completes at time $T$.  
By  the property of the
sets  $P_1,\dots,P_n$, all  $n$  machines are now completely filled up to time
$T+{\OPT}$. Figure~\ref{fig: ideal schedule} shows an
illustration of such a schedule.


\begin{figure}[htbp]
\begin{center}
\includegraphics[width=6.2in]{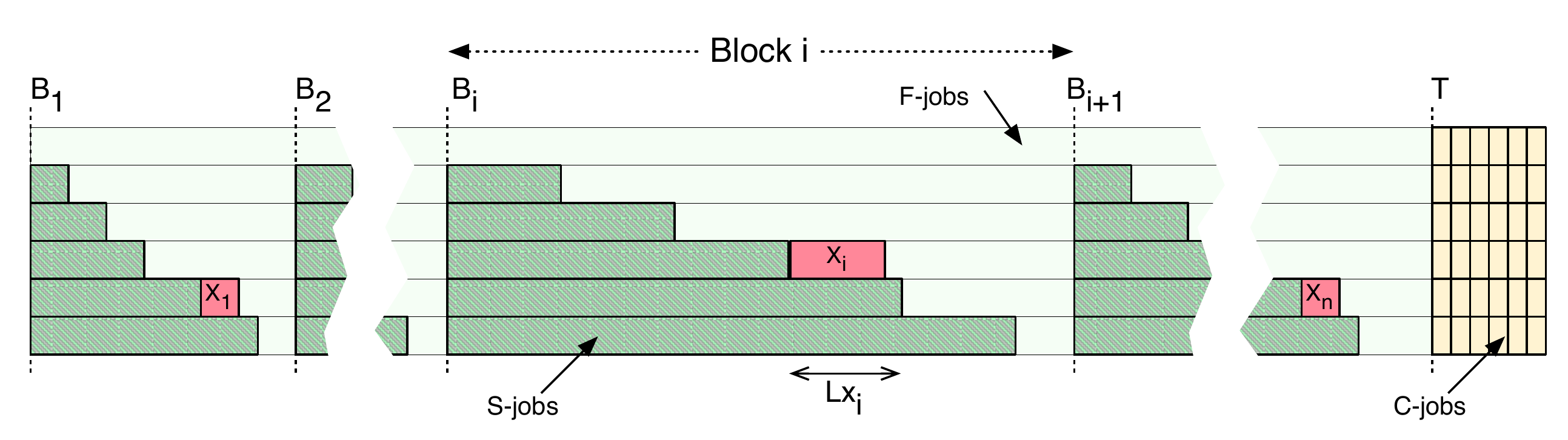}
\caption{A schedule with delay equal to ${\OPT}$.}
\label{fig: ideal schedule}
\end{center}
\end{figure}


We now  focus on the delay of  this schedule. The $\rmS$-jobs  and $\rmC$-jobs do
not contribute to the delay, as they start at their release time. The $\rmF$-jobs
complete at time $T$ and hence generate a delay of $nyL\frac{n(n-1)}{2}+ \lambda  ny$.  
Finally, all  $\rmX$-jobs that complete  on a machine $k$
generate a  total delay of $(k-1)Ly$. Adding up the delays of different types of
jobs, we obtain the total delay of {\OPT}, as required.

\medskip
$(\Leftarrow)$
In the following,  the ``if'' direction of the proof is detailed.  As a
first step,  we fix  a schedule  $\sigma$ of delay  at most  ${\OPT}$. 
To simplify the argument, we first argue that we can make some simplifying 
assumptions about $\sigma$.

We first assume that $\sigma$ is $\OPT$-dominant, where the dominance property is defined as follows.
For two schedules $\sigma_1$, $\sigma_2$, we say that $\sigma_1$ \emph{strictly dominates}
$\sigma_2$  if either
\begin{description}
\item{(i)} there  is  a time  $t_0$  such that  in  all time  slots
$t<t_0$, $\sigma_1$ and $\sigma_2$  schedule the same number of $\rmS$-jobs,
while at  time $t_0$, $\sigma_1$  schedules strictly more  $\rmS$-jobs than
$\sigma_2$, or
\item{(ii)} $\sigma_1$ and $\sigma_2$ execute the same number of $\rmS$-jobs in each
time slot, and there is  a time $t_0$ such that  in  all time  slots
$t<t_0$, $\sigma_1$ and $\sigma_2$  schedule the same number of $\rmX$-jobs,
while at  time $t_0$, $\sigma_1$  schedules strictly more  $\rmX$-jobs than
$\sigma_2$. (As we shall prove shortly, we can assume that
at most one $\rmX$-job is executed at any time.)
\end{description}
Then $\sigma$ is called $\OPT$-dominant if it is not strictly dominated
by any other schedule of delay at most $\OPT$. Clearly, there may be many
$\OPT$-dominant schedules.

In   addition  we  assume   the  \emph{vertical  ordering
  property}, stating that in every  time slot, jobs are sorted from
the first  to the  last machine according to  the order  
$F_1,\dots,F_n$, $X_1,\dots ,X_{3n}$, followed by $\rmS$-jobs and then followed by $\rmC$-jobs.
Reordering  the units of jobs scheduled at the same time slot on different machines
has no impact  on the delay of
the  schedule, so this  last assumption  can be  made without loss of
generality.


\begin{Lemma}\label{lemme:block}
In $\sigma$, all  $\rmX$-jobs and all $\rmS$-jobs complete  strictly before  the end  of their
blocks, and the $\rmF$-jobs complete no later than at time $T$.
\end{Lemma}

\begin{proof}
Since  the total delay is assumed  to be at most  {\OPT}, and no jobs are released
in the last slots of each block, all $\rmS$-jobs
and $\rmX$-jobs have  to complete strictly before  the end  of their corresponding blocks. 
  
To show the second claim, without loss of generality, we assume that $\rmC$-jobs are scheduled
in order of release times, that is, for any two $\rmC$-jobs $i$ and $j$,
$r_i < r_j$ implies that $j$ is scheduled not earlier than $i$.
Suppose that some $\rmF$-job completes at time $T+\tau$, where $\tau>0$. 
If $\tau\ge\OPT+1$, then the delay of this $\rmF$-job exceeds $\OPT$.
Otherwise, using our assumption about $\rmC$-jobs, this $\rmF$-job
forces a delay for at last one $\rmC$-job released in each time slot of the cork block. Since the $\rmF$-job has itself a
delay of at least $\tau$, so the total delay would also exceed $\OPT$. 
\qed \end{proof}


\begin{Lemma}\label{lemme:idle}
There is no idle time in $\sigma$ before time $T$.
\end{Lemma}

\begin{proof}
  Using Lemma~\ref{lemme:block}, all $\rmS$-jobs, $\rmX$-jobs and $\rmF$-jobs have to
  be finished at  time $T$.  Since the total  processing time of these
  jobs is $nT$, the lemma follows.
\qed \end{proof}


\begin{Lemma}\label{lemme:complete}
In $\sigma$, all $\rmF$-jobs complete exactly at time $T$.
\end{Lemma}

\begin{proof}
At  time $T-1$,  all  $\rmS$-jobs  and $\rmX$-jobs  are  already completed,  by
Lemma~\ref{lemme:block}.   So   only  $\rmF$-jobs can execute at time $T-1$. 
The lemma now follows from Lemma~\ref{lemme:idle}.
\qed \end{proof}

As previously observed, we can assume that the $\rmC$-jobs do  not generate any
delay. By Lemma~\ref{lemme:complete}, the total delay of $\rmF$-jobs
is $nyL\frac{n(n-1)}{2} + \lambda ny$.
Removing  this contribution of $\rmF$-jobs to  the objective function
leads to the following bound, which will play an important role in the rest of the proof:

\begin{equation}
\delay_{\rmS}(\sigma) + \delay_{\rmX}(\sigma) \leq yL\frac{n(n-1)}{2}.
\label{eqn:remaining}
\end{equation}


The lemma below says that the $\rmS$-jobs of each block must be scheduled in $\sigma$ so as to form
a staircase shape, similar (but not necessarily identical)
to the schedule shown in Figure~\ref{fig: ideal schedule}.

\begin{Lemma}\label{lem:stair}
Schedule $\sigma$ has the following property: in each block, going from left to right,
the  numbers  of  $\rmS$-jobs in the time slots of this block form a non-increasing sequence.
\end{Lemma}

\begin{proof}
The proof uses an argument by contradiction. Suppose that
  inside some block  there are time slots $t-1$ and $t$  such that at $t$
  there are  strictly more  $\rmS$-jobs scheduled than  at $t-1$.   For the
  purpose of  this proof, we may  assume that $\rmS$-jobs  are scheduled in
  order of their release times,  since $\delay_{\rmS}$ is independent of the
  ordering of the $\rmS$-jobs.  Therefore  one of the $\rmS$-jobs from time slot
  $t$ must be  already released at time $t-1$,  by the release pattern
  of the $\rmS$-jobs.   Call this job $\ell$.  At time $t-1$  no machine can be
  idle, by Lemma~\ref{lemme:idle}.   So at time $t-1$ we have more units of $\rmX$- or
$\rmF$- jobs than at time $t$. This implies that there is a job $\ell'$
of type $\rmX$ or $\rmF$ scheduled at time $t-1$ that is not scheduled at time $t$.
Hence we
  can exchange  these units of $\ell$ and $\ell'$ without increasing  the total delay.
But this  contradicts the assumption that $\sigma$ is {\OPT}-dominant.
\qed \end{proof}


\begin{Lemma}\label{lem:job order}
Each machine, within each block, executes units of jobs in the following
order: first $\rmS$-jobs, then $\rmX$-jobs (if any), and finally $\rmF$-jobs.
\end{Lemma}

\begin{proof}
Fix some block $i$. By Lemma~\ref{lem:stair}, the slots occupied by 
$\rmS$-jobs in this block form a staircase shape; thus on each machine they
are executed before $\rmX$-jobs and $\rmF$-jobs. We still need to show that
slots occupied by $\rmX$-jobs precede those occupied by $\rmF$-jobs.
We know that only one $\rmX$-job is executed in this block, namely $X_i$.

We argue by contradiction. Suppose that there is a time $t$ in this block
such that some machine $k$ executes $X_i$ at time $t$ and executes some
$\rmF$-job at time $t-1$. At both times, according to vertical ordering 
property of $\sigma$, all machines $1,...,k-1$ execute
$\rmF$-jobs and machines $k+1,...,n$ execute $\rmS$-jobs.
This implies that there is some $\rmF$-job executed at time $t-1$ that is
not executed at time $t$. We can then exchange this $\rmF$-slot with
the slot of $X_i$ at time $t$, without increasing the delay 
(by Lemma~\ref{lemme:complete}), thus obtaining a contradiction with $\OPT$-dominance of $\sigma$.
\qed	
\end{proof}

Lemmas~\ref{lemme:block} to \ref{lem:job order} characterize
the structure of our schedule $\sigma$:

\begin{enumerate}
\item All $\rmC$-jobs are scheduled at their release times.
\item All $\rmS$-jobs and $\rmX$-jobs complete in their respective blocks.
\item The $\rmF$-jobs fill in the remaining slots in the blocks and they complete exactly at time $T$.
\item In each block $i$, for each $k$, the number of $\rmS$-jobs executed by machine $k$ is non-decreasing with $k$;
that is, the $\rmS$-jobs form a staircase pattern in each block.
\item In each block $i$, each machine $k$ first executes some $\rmS$-jobs, then some number of units of $X_i$ (if any),
and then some units of $\rmF$-jobs.
\end{enumerate}


\begin{figure}[htbp]
\begin{center}
\includegraphics[width=4.2in]{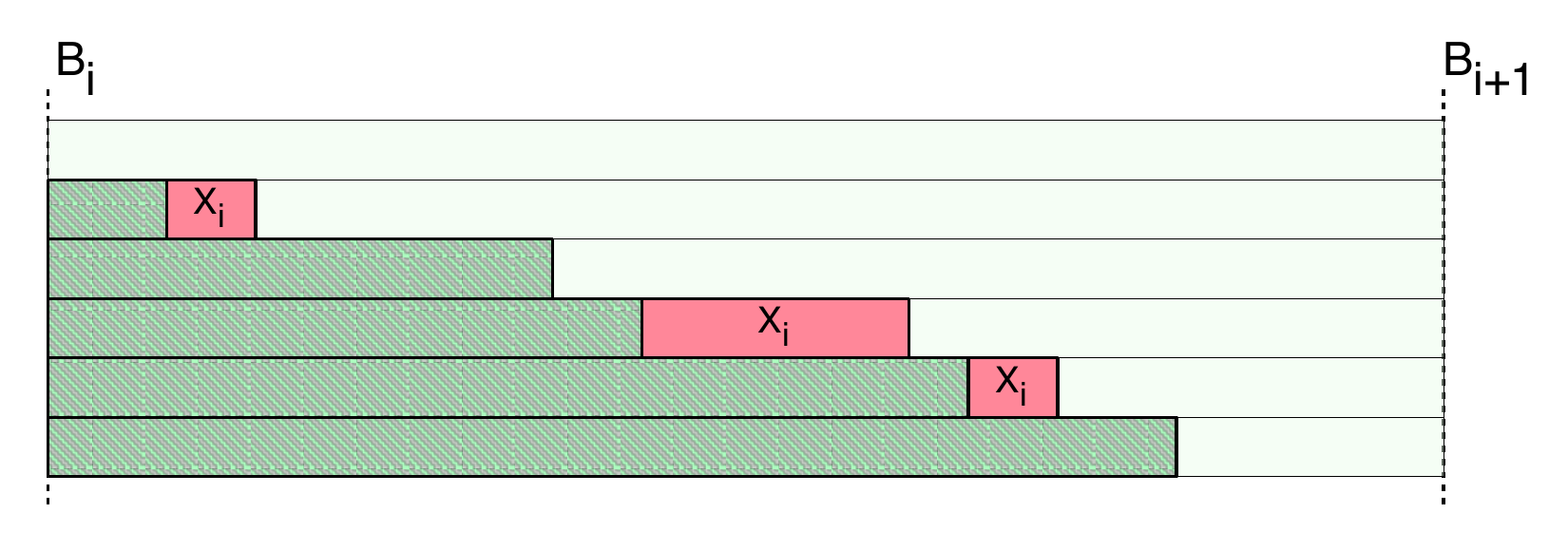}
\caption{The structure of a block in schedule $\sigma$.}
\label{fig: typical block}
\end{center}
\end{figure}


Thus the overall structure of $\sigma$ is similar to the schedule shown in 
Figure~\ref{fig: ideal schedule}; however, within each block $i$ the steps of the staircases of $\rmS$-jobs
may have different lengths, and $X_i$ may not be scheduled in one contiguous block on one machine.
A typical structure of a block is shown in Figure~\ref{fig: typical block}.

\medskip

For each machine $k$ and each block $i$,
the \emph{ideal number} of $\rmS$-jobs on machine $k$ in block $i$ is defined to be
$(k-1)Lx_i$.  This value corresponds to the  number of $\rmS$-jobs
processed on  machine $k$ in  block $i$ in  the $(\Rightarrow)$ direction of
the proof (see Figure~\ref{fig: ideal schedule}).
We also define $\excess_k^i$ to be the \emph{excess} of $\rmS$-jobs on machine
$k$ in block $i$, namely the difference between the actual and
the ideal number of $\rmS$-jobs on machine $k$ in block $i$, if this difference
is non-negative; otherwise let we $\excess_k^i$ be $0$. 
More specifically, if the number of $\rmS$-jobs scheduled by
$\sigma$ on machine $k$ in block $i$ is $c$ then
$\excess_k^i = \max\braced{ c - (k-1)Lx_i, 0}$.


\begin{Lemma}\label{lemme:Dcost}
In schedule $\sigma$, the total delay of the $\rmS$-jobs is bounded by
the following inequality: \\
$\delay_{\rmS}(\sigma) \geq \sum_i \sum_k Lx_i\excess_k^i$.
\end{Lemma}

\begin{proof}
  For every time
  $t$,  let $z(t)$  be  the  difference between  the  number of  $\rmS$-jobs
  released  at times  $0,1,...,t$  and  the  number of  $\rmS$-jobs
  scheduled at the same times in $\sigma$.  Clearly we have $z(t)\geq 0$ for every
  time $t$.

  Consider an  $\rmS$-job $j$  with release  time $r_j$  and  completion time
  $C_j$. Its delay  is $D_j = C_j-1-r_j$. Job $j$ contributes one unit
  to  $z(t)$ for each time $t = r_j,r_j+1,...,C_j-2$;  hence $\delay_{\rmS}(\sigma) =\sum_t  z(t)$.  

  Now fix a block $i$ and machine $k\geq 2$ with $\excess_k^i > 0$, and let $t_k^i=B_i
  + (k-1)Lx_i$. In the interval $[t_k^i,t_k^i+\excess_k^i)$ at least $(n-k+1)\excess_k^i$ $\rmS$-jobs
  are  scheduled,  but  only  $(n-k)\excess_k^i$  $\rmS$-jobs  are  released.   The
  additional scheduled $\rmS$-jobs must therefore be released strictly before
  time $t_k^i$, which implies $z(t_k^i-1) \geq \excess_k^i$.

	At each time $\tau = t_k^i-Lx_i,...,t_k^i-1$ the number of released $\rmS$-jobs
	is $n-k+1$. By the choice of $i$ and $k$, and by Lemma~\ref{lem:stair},
	at each time
	$\tau$ in this range the number of executed $\rmS$-jobs is at least $n-k+1$.
  We thus obtain that $z(\tau)$ is non-increasing in this range, so we have
\begin{equation*}
	z(t_k^i-Lx_i)\geq  \ldots \geq z(t_k^i-1) \geq \excess_k^i,
\end{equation*}
and thus
\begin{equation*}
	 z(t_k^i-Lx_i) + \ldots + z(t_k^i-1) \geq Lx_i \excess_k^i,
\end{equation*}
which concludes the  proof of the lemma by summing over all choices of $i$ and $k$.
\qed \end{proof}

We  complete  the  proof  with  an  upper bound  on  the  excess  values,
that motivates our choice for the value of parameter $\lambda$.


\begin{Lemma}\label{eqn:epsilonD}
We have $\sum_i \sum_k \excess_k^i < \lambda$.
\end{Lemma}
\begin{proof}
  For a  proof by contradiction assume  that the left hand  side is at
  least  $\lambda$.   
Then Lemma~\ref{lemme:Dcost},  together with  the assumption  that
  $x_i > y/4$ for all $i$,  would imply  that
\begin{equation*}
\delay_{\rmS}(\sigma) 
  				> \frac{yL}{4} \sum_i \sum_k \excess_k^i \geq \frac{yL}4 \lambda=yLn(n-1)/2.
\end{equation*}
This, however, contradicts the bound in Equation~\eqref{eqn:remaining}.

\qed 
\end{proof}

We denote  by $P_k$  the set of indices of  the $\rmX$-jobs completing  on a machine
$k$, independent of whether they are scheduled entirely on machine $k$ or
not. By the construction,  the lengths of all  $\rmX$-jobs   are multiples of
$\lambda$, which is a strict upper bound on the number of $\rmS$-jobs that can diverge from
the ideal pattern (by Lemma~\ref{lemme:Dcost}).
We will use  this fact, to show that the partition
$P_1,\ldots,P_n$ forms a solution to the original instance of  $\ThreePartition$.


\begin{Lemma}\label{lemma:qtyE}
The following inequality holds for all machines $k$: 
\begin{equation*}
    \sum_{i\in P_k\cup\ldots\cup P_n}x_i \geq (n-k+1)y.
\end{equation*}
\end{Lemma}

\begin{proof}
The amount of $\rmS$-jobs executed by each machine $\ell$ can
be bounded by the amount of $\rmS$-jobs in the ideal pattern (from the proof of the ``if'' implication)
plus the excess values for all blocks, which works out to be at most
$(\ell-1)Lny + \sum_i E^i_\ell$. Therefore,
by Lemma~\ref{eqn:epsilonD}, the total amount of $\rmS$-jobs executed on machines 
numbered $k,...,n$ is strictly smaller than
\begin{equation*}
    (k-1)Lny+\ldots+(n-1)Lny+\lambda.
\end{equation*}
By  the vertical  ordering assumption,  each job  $F_\ell$ can  only be
scheduled on  machines $1,...,\ell$  and therefore the  total amount of
$\rmF$-jobs on machines numbered $k,...,n$ is at most
\begin{align*}
[\,T-(k-1)Lny-\lambda y \,] + ... + [\, T-(n-1)Lny-\lambda y \,].
\end{align*}
Together, the total amount of $\rmS$-jobs and $\rmF$-jobs on machines
$k,...,n$ is strictly smaller than $(n-k+1)T -  (n-k+1)\lambda y + \lambda$.
Thus, using Lemma~\ref{lemme:idle},
the total  length of $\rmX$-jobs completing on machines $k,...,n$ satisfies
\begin{equation*}
	\sum_{i\in P_k\cup\ldots\cup P_n}x_i >  (n-k+1)\lambda y - \lambda.
\end{equation*}

The lemma follows from the fact  that every $\rmX$-job has length that
is a multiple of $\lambda$.
\qed \end{proof}

To complete  the proof  of the theorem,  we focus  on the delay  of the $\rmX$-jobs.   

We know  that if some machine $k$ executes at least one unit of a job
$X_i$ in block $i$, then all earlier slots of machine $m$ in this block execute $\rmS$-jobs. 
Therefore the amount of $\rmS$-jobs on the machines and in blocks where $\rmX$-jobs
complete can give us
a lower bound  on the total completion time of the $\rmX$-jobs, from  which we have to
subtract the total processing time to obtain a bound on $\delay_{\rmX}(\sigma)$.  

In the ideal schedule, the delay of $\rmX$-jobs that complete
on a machine $k$ is $(k-1)Ly$, but in
$\sigma$ this delay could be different. To get a lower bound on this delay,
define the \emph{deficiency of $\rmS$-jobs on machine $k$ in a block $i$} to be
the ideal value of $(k-1)Lx_i$ minus the number of $\rmS$-jobs executed by $k$
in block $i$ of schedule $\sigma$, if this value is non-negative; otherwise we 
let the deficiency to be $0$. Then the delay of $\rmX$-jobs that complete
on a machine $k$
is at least $(k-1)L\sum_{i\in P_k} x_i$ minus the total deficiency of 
$\rmS$-jobs on machine $k$.
But the total deficiency of $\rmS$-jobs, overall all machines and all blocks,
is the same as the total excess $\sum_i\sum_k E^i_k$, so it is strictly
smaller than $\lambda$, by Lemma~\ref{eqn:epsilonD}.
This leads to the lower bound
\begin{align}
  \delay_{\rmX}(\sigma) & > L \sum_{i\in P_2}x_i + 2  L \sum_{i\in P_3}
  x_i +\ldots+(n-1)L\sum_{i\in P_n} x_i - \lambda - n\lambda y,
	\label{eq: delay of x-jobs}
\end{align}
where we use the fact that the total processing time of the $\rmX$-jobs is $n\lambda y$.

Putting everything together now, we apply
Inequality~(\ref{eqn:remaining}), Inequality~(\ref{eq: delay of x-jobs}) above, 
substitute $L = ny\lambda + \lambda$, and then
apply Lemma~\ref{lemma:qtyE} (summing over all values of $k$), 
obtaining the following sequence of inequalities:
\begin{align*}
	yL\frac{n(n-1)}{2} &\ge  \delay_{\rmX}(\sigma) 
		\\
		& > L\left( \sum_{i\in P_2\cup \ldots \cup P_n} x_i
						+ \sum_{i\in P_3\cup \ldots \cup P_n} x_i
    					+ \ldots 
							+\sum_{i\in P_n} x_i\right) - L
		\\
		&\ge yL\frac{n(n-1)}{2} - L.
\end{align*}
The first and last expressions are multiples of $L$. From this derivation we
can thus conclude that the bound from Lemma~\ref{lemma:qtyE}, used in the
last inequality, must be in fact tight.
Specifically, we get that
$\sum_{i\in P_k\cup \ldots \cup P_n}x_i = (n-k+1)y$ for each $k = 2,...,n$.
We thus obtain that $\sum_{i\in P_1} x_i = y$ and, proceeding
by induction on $k$, we have $\sum_{i\in P_k} x_i = y$ as well for all other
machines $k$. This completes the proof of Theorem~\ref{thm:main}.


\paragraph{Acknowledgements.}
M.~Chrobak was partially supported by National Science Foundation grant CCF-1217314.
The authors would like to thank anonymous reviewers for many useful comments and suggestions.

\bibliographystyle{plainnat}
\bibliography{Biblio}
\end{document}